\documentclass{amsart}
\usepackage{amsmath}
\usepackage{amsthm}
\usepackage{amsfonts}
\usepackage{amssymb}
\usepackage{graphicx}

\newcommand{\tmop}[1]{\ensuremath{\operatorname{#1}}}

\newcommand{\nocomma}{}

\newcommand{\tmstrong}[1]{\textbf{#1}}

\theoremstyle{plain}
\newtheorem{theorem}{Theorem}
\newtheorem{proposition}{Proposition}
\newtheorem{lemma}{Lemma}
\newtheorem{conjecture}{Conjecture}
\newtheorem{definition}{Definition}
\newtheorem*{remark}{Remark}

\begin{document}
\title{A new large class of functions not APN infinitely often}
\author{Florian Caullery}

\thanks{ Institut de Math{\'e}matiques de Luminy, CNRS-UPR9016, 163 av. de Luminy,
case 907, 13288 Marseille Cedex 9, France.}
\thanks{Email: {\tt florian.caullery@etu.univ-amu.fr}}
        
\date{\today}

\maketitle

\begin{abstract}
 In this paper, we show that there is no vectorial Boolean function of degree
  $4 e$, with $e$ satisfaying certain conditions, which is APN over infinitely many extensions of its field of definition. It is a new step in the proof of the conjecture of Aubry, McGuire and Rodier.
\end{abstract}

\section{Introduction}
\label{intro}

A vectorial Boolean function is a function $f: \mathbb{F}_{2^m} \rightarrow \mathbb{F}_{2^m}$. This object arises in fields like cryptography and coding theory and is of particular interest in the study of block-ciphers using a substitution-permutation network (SP-network) since they can represent a Substition Box (S-Box). In 1990 Biham and Shamir introduced the differential cryptanlysis in \cite{Biham1990}. The basic idea is to analysis how a difference between two inputs of an S-box will influence the difference between the two outputs. This attack was the motivation for Nyberg to introduce the notion of Almost Perfectly Nonlinear (APN) function \cite{Nyberg1994} which are the function providing the S-Boxes with best resistance to the differential cryptanalysis. An APN function is a vectorial Boolean function such that $\forall a \neq 0, b \in \mathbb{F}_{2^m}$ there exist at most two solutions to the equation:
\[
f(x+a)+f(x)=b
\]
The problem of the classification of all APN functions is challenging and has been studied by many authors. In a first time, the studies focused on power functions and it was recently extended to polynomial functions (Carlet, Pott and al \cite{Carlet1998,Edel2006,Edel2009}) or polynomials on small fields (Dillon \cite{Dillon2009}). On the other hand, several authors (Berger, Canteaut, Charpin, Laigle-Chapuy \cite{Berger2006}, Byrne, McGuire \cite{Byrne2008} or Jedlicka \cite{Jedlicka2007}) showed that APN functions cannot exist in certain cases. Some also studied the APN functions on fields of odd characteristic (Leducq \cite{Leducq2012}, Pott and al. \cite{Dobbertin2003,Poinsot2011}, Ness, Helleseth \cite{Ness2007} or Wang, Zha \cite{Zha2010,Zha2011} ). \\
One way to approach the problem of the classification is to consider the function APN over infinitely many extensions of $\mathbb{F}_2$, namely, the exceptional APN functions. The two best known exceptional APN functions are the Gold functions: $f(x)=x^{2^i+1}$ and the Kasami functions $f(x)=x^{4^i-2^i+1}$, both are APN whenever $i$ and $m$ are coprime. We will refer to $2^i+1$ and $4^i-2^i+1$ respectively as the Gold and Kasami exponent. It was proved by Hernando and McGuire in \cite{Hernando2011} that those two functions are the only monomial exceptional APN functions. It was the starting point for Aubry, McGuire and Rodier to formulate the following conjecture:

\begin{conjecture}(\cite{Aubry2010})
The only exceptional APN functions are, up to Carlet Charpin Zinoviev-equivalence (as defined below), the Gold and Kasami functions. 
\end{conjecture}

We provide the definition of the Carlet Charpin Zinoviev equivalence:

\begin{definition}(\cite{Carlet1998})
Two functions $f$ and $g$ are Carlet Charpin Zinoviev (CCZ-)equivalent if there exist a linear permutation between their graphs (i.e. the sets $\{x,f(x)\}$ and $\{x,g(x)\}$).
\end{definition}

It has to be noted that all the functions CCZ-equivalent to an APN function are also APN \cite{Carlet1998}.

By means of a simple rewriting of the definition of APN function in terms of algebraic geometry, Rodier was able to prove that, if the projective closure of the surface $X$ defined by the equation:

\[
\frac{f(x)+f(y)+f(z)+f(x+y+z)}{(x+y)(y+z)(z+x)}=0
\]

has an absolutely irreducible component defined over $\mathbb{F}_{2^m}$, then $f$ is not an exceptional APN function \cite{Rodier2008}. The idea now is to exploit this criteria to prove that the functions which are not CCZ-equivalent to a Gold or Kasami function are not exceptional APN. This approach enabled Aubry, McGuire and Rodier to state, for example, that there is no exceptional APN function of degree odd not a Gold or Kasami exponent and of degree $2e$ with $e$ an odd number \cite{Aubry2010}. 

From now on we let $q=2^m$,

\[
\phi(x,y,z)=\frac{f(x)+f(y)+f(z)+f(x+y+z)}{(x+y)(y+z)(z+x)}
\]
and
\[
\phi_i (x,y,z)=\frac{x^i+y^i+z^i+(x+y+z)^i}{(x+y)(y+z)(z+x)}
\]

In this paper we continue in the same way than Aubry, McGuire and Rodier and are interested in the functions of degree $4e$ with $e$ such that $\phi_e$ is absolutely irreducible. As shown by Janwa and al. (\cite{Janwa1993} and \cite{Janwa1995}) it is the case for example when $e \equiv 3 \pmod{4}$ or when $e \equiv 5 \pmod{8}$ and the maximum cyclic code of length $\frac{e-1}{4}$ has no codewords of weight 4. In particular, $e$ cannot be a Gold or a Kasami exponent. There are many others $e$ which satisfy the condition. It was even conjectured that it was the case of any $e$ odd not a Gold or Kasami exponent but $e=205$ was shown to be the smallest counter-example by Hernando and McGuire \cite{Hernando2011}. We now give an overview of the classification of the exceptional APN function.

\section{The state of the art}
\label{sec:1}

Using the approach described in the introduction Aubry, McGuire and Rodier obtained the following results in \cite{Aubry2010}.

\begin{theorem}
  {\tmstrong{(Aubry, McGuire and Rodier, \cite{Aubry2010})}} If the degree of the
  polynomial function $f$ is odd and not an exceptional number then $f$ is not an exceptional APN function.
\end{theorem}

\begin{theorem}
  {\tmstrong{(Aubry, McGuire and Rodier \cite{Aubry2010})}} If the degree of the polynomial function $f$ is $2e$ with e odd and if $f$ contains a term of odd degree, then $f$ is not an exceptional APN function.
\end{theorem}

There are some results in the case of Gold degree $2^i + 1$:

\begin{theorem}
  {\tmstrong{(Aubry, McGuire and Rodier \cite{Aubry2010})}} Suppose $f \left( x \right) =
  x^{2^i+1} + g \left( x \right)$ where $\deg \left( g \right) \leqslant 2^{i - 1} +
  1$. Let $g \left( x \right) = \sum_{j = 0}^{2^{i - 1} + 1} a_j x^j$. Suppose
  moreover that there exists a nonzero coefficient $a_j$ of $g$ such that
  $\phi_j \left( x, y, z \right)$ is absolutely irreducible. Then $f$ is not an exceptional APN function.
\end{theorem}

This result has been consequently extended by Delgado and Janwa in \cite{Delgado2012} with the two following theorems:
\begin{theorem}
 {\tmstrong{(Delgado and Janwa \cite{Delgado2012})}}
For $k \geq 2$, let $f(x)=x^{2^i+1}+h(x) \in \mathbb{F}_q$ where $\text{deg}(h) \equiv 3  \pmod{4} < 2^i+1.$ Then $f$ is not an exceptional APN function. 
\end{theorem}

and

\begin{theorem}
 {\tmstrong{(Delgado and Janwa \cite{Delgado2012})}}
For $k \geq 2$, let $f(x)=x^{2^i+1}+h(x) \in \mathbb{F}_q$ where $\text{deg}(h)=d \equiv 1  \pmod{4}  < 2^i+1.$ If $\phi_{2^i+1},\phi_d$ are relatively prime, then $f$ is not an exceptional APN function. 
\end{theorem}

There also exist a result for polynomials of Kasami degree $ 2^{2 i} - 2^i + 1 $:

\begin{theorem}
  {\tmstrong{(F\'erard, Oyono and Rodier \cite{Ferard2012})}} Suppose $f \left( x \right) =
  x^{2^{2 i} - 2^i + 1 } + g \left( x \right)$ where $\deg \left( g
  \right) \leqslant 2^{2 k - 1} - 2^{k - 1} + 1$. Let $g \left( x \right) =
  \sum_{j = 0}^{2^{2 k - 1} - 2^{k - 1} + 1} a_j x^j .$ Suppose moreover that
  there exist a nonzero coefficient $a_j$ of g such that $\phi_j \left( x, y,
  z \right)$ is absolutely irreducible. Then $f$ is not an exceptional APN function.
\end{theorem}

Rodier proved the following results in \cite{Rodier2011}. We recall that for any function
$f : \mathbb{F}_q \rightarrow \mathbb{F}_q$ we associate to $f$ the
polynomial $\phi \left( x, y, z \right)$ defined by:
\[ \phi \left( x, y, z \right) = \frac{f \left( x \right) + f \left( y \right)
   + f \left( z \right) + f \left( x + y + z \right)}{\left( x + y \right)
   \left( x + z \right) \left( y + z \right)} . \]
\begin{theorem}
  {\tmstrong{(Rodier \cite{Rodier2011})}} If the degree of a polynomial function $f$ is such that $\deg \left( f \right) = 4 e$ with $e \equiv 3  \pmod{4}
  $, and if the polynomials of the form \[ \left( x + y \right) \left( x + z \right) \left( y + z \right) + R, \] with
  \[ R \left( x, y, z \right) = c_1 \left( x^2 + y^2 + z^2 \right) + c_4
     \left( xy + xz + zy \right) + b_1 \left( x + y + z \right) + d_1, \]
  for $c_1, c_4, b_1, d \in \mathbb{F}_{q^3}$, do not divide $\phi$, then $f$ is not an exceptional APN function.
\end{theorem}

There are more precise results for polynomials of degree 12.

\begin{theorem}
  {\tmstrong{(Rodier \cite{Rodier2011})}} If the degree of the polynomial $f$ defined over $\mathbb{F}_q$ is 12, then either $f$ is not an exceptional APN function or $f$ is CCZ-equivalent to the Gold function $x^3$.
\end{theorem}

\section{Our main Result}
\label{sec:2}

The goal of this paper is to prove the following result: 

\begin{theorem}
Let $f: \mathbb{F}_q \rightarrow \mathbb{F}_q$ of degree $4e$ with $e>3$ such that $\phi_e$ is absolutely irreducible. Then $f$ is not an exceptional APN function.
\end{theorem}

The proof of this theorem is decomposed in two main steps. The first one is to show that the exceptional APN functions of degree as in the conditions of theorem 9 must be of a certain form. The second one is to prove that they are hence CCZ-equivalent to a nonexceptional APN function, which is a contradiction.

\section{The divisibility condition}
\label{sec:3}
In the statement of theorem 7 in \cite{Rodier2011} the condition that $e$ must be  $3 \pmod{4}$ is only used to guarantee that $\phi_e$ is absolutely irreducible  (as shown in \cite{Janwa1993}). It is easy to see that the proof works whenever $e$ is such that $\phi_e$ is absolutely irreducible. As a consequence of this remark theorem 7 can be directly extended as follow:

\begin{theorem} \label{divisionCondition}
Let $f: \mathbb{F}_q \rightarrow \mathbb{F}_q$ be of degree $d = 4 e$ with $e$ such that $\phi_e$ is absolutely irreducible. If the polynomials of the form
  \[ \left( x + y \right) \left( x + z \right) \left( y + z \right) + R(x,y,z), \]
  with
  \[ R \left( x, y, z \right) = c_1 \left( x^2 + y^2 + z^2 \right) + c_4
     \left( xy + xz + zy \right) + b_1 \left( x + y + z \right) + d, \]
  for $c_1, c_4, b_1, d_1 \in \mathbb{F}_{q^3}$, does not divide $\phi$ then $f$
  is not an exceptional APN function.
\end{theorem}

\begin{remark}
As said in the introduction, $\phi_e$ is absolutely irreducible in many cases including $e \equiv 3 \pmod{4}$.
\end{remark}

\begin{remark}
Among the examples where $\phi_e$ is not absolutely irreducible, we would like to draw attention on two particular cases. Firstly, one can quickly verify that $\phi_e$ is not irreducible when $e$ is even (see \cite{Aubry2010} lemma 2.2). Secondly, when $e$ is a Gold or a Kasami exponent there exists a decomposition of $\phi_e$ into absolutely irreducible factors (see \cite{Janwa1993}).
\end{remark}

We will now investigate the consequences of the last theorem.

Let $f : \mathbb{F}_q \rightarrow \mathbb{F}_q$ be a function of degree $d=4
e$ where $e >3 $ is odd and such that $\phi_e$ is absolutely irreducible. Suppose now that $f$ is an exceptional APN function. We recall that
\[ \phi \left( x, y, z \right) = \frac{f \left( x \right) + f \left( y
   \right) + f \left( z \right) + f \left( x + y + z \right)}{\left( x + y
   \right) \left( y + z \right) \left( z + x \right)}, \]
Writing $f\left( x \right)=\sum_{i=0}^d a_i x^i$ we have
\[ \phi_f = \sum_{i = 0}^d a_i \phi_i, \]
We can fix $a_d$ to 1 without loss of generality as $\mathbb{F}_q$ is a field.

Let $\rho $ be a generator of the Galois group $ \tmop{Gal} \left( \mathbb{F}_{q^3} /\mathbb{F}_q. \right)$ and let us consider \linebreak $c_1, c_4, b_1, d_1 \in \mathbb{F}_{q^3}$, $R \left( x, y, z \right)
= c_1 \left( x^2 + y^2 + z^2 \right) + c_4 \left( xy + xz + zy \right) + b_1
\left( x + y + z \right) + d$ and $A = \left( x + y \right) \left( y + z
\right) \left( z + x \right)$. \\

 As a consequence of theorem \ref{divisionCondition}, we may assume that the polynomial \linebreak $P=\left( A + R \right) \left( A + \rho \left( R \right) \right) \left( A + \rho^2 \left( R \right) \right)$ divides $\phi$. We denote $P_i$ the homogeneous component of degree $i$ of $P$.
As $\phi$ is of total degree $d - 3$, there exists a polynomial $Q \in
\mathbb{F}_{q^3} \left[ x, y, z \right]$ of total degree $d - 12$ such that $\phi=P \times Q$. Denoting $Q_i$ the homogeneous component of $Q$ of degree $i$ we get
\[ \sum_{i = 0}^9 P_i \cdot \sum_{i = 0}^{d - 12} Q_i = \sum_{i = 0}^d a_i
   \phi_i . \]
As $\phi$ is a symmetrical polynomial in $x, y, z$ we can write it using symmetrical functions $s_1 = x + y + z$, $s_2 = xy + xz + yz$ and $s_3 = xyz$ (see \cite{Bourbaki} chapter 6).
Denoting $p_i = x^i + y^i + z^i$, we have $p_i = s_1 p_{i - 1} + s_2 p_{i - 2}
+ s_3 p_{i - 3} \nocomma$. We remark that $\phi_i = \frac{p_i + s_1^i}{A}$ and
that $A = \left( x + y \right) \left( y + z \right) \left( z + y \right) = s_1
s_2 + s_3$.

We shall now determine all the coefficients of $R$ identifying degree by degree $P$, $Q$ and $\phi$.

\begin{proposition}
  If $A+R$ divides $\phi_f$, then $R = c_1 \phi_5 + c_1^3$ and the trace of $c_1$
  in $\mathbb{F}_{q^3}$ is 0. Moreover the polynomial $\left( A + R \right) \left( A
+ \rho \left( R \right) \right) \left( A + \rho^2 \left( R \right) \right)$ is equal to 
\[
\frac{L \left( x \right)^3 + L \left( y \right)^3 + L \left( z \right)^3 +
   L \left( x + y + z \right)^3}{\left( x + y \right) \left( y + z \right)
   \left( z + x \right)}
\]
where $L \left( x \right) = x \left( x + c_1 \right) \left( x + \rho\left( c_1 \right) \right) \left( x + \rho^2 \left( c_1 \right) \right)$.
\end{proposition}

\begin{proof}

We will need the following lemmas : 

\begin{lemma}\label{e3mod4}
Suppose $e \equiv 3  \pmod{4}$ and let $s=x+y$. We have :
\[
\left( x+z \right)^2 \phi_e = \left( x^{e-1}+z^{e-1} \right) + s\frac{ \left( x^{e-2}z+z^{e-2}x \right)}{x+z}+s^2\frac{\left( x^{e-3}+z^{e-3} \right) \left( x^2+z^2+xz \right)}{\left( x+z \right)^2}  \pmod{s^3} 
\]
\end{lemma}

\begin{proof}
We have 
\[
A \phi_e=x^e+y^e+z^e+(x+y+z)^e.
\]
Let us put $s=y+z$. We get

\begin{eqnarray*}
  &  & \left( x + z \right) \left( s + x + z \right) s \phi_e\\
  & = & x^e +\left( s + z \right)^e + z^e + \left( x + s \right)^e\\
  & = & s \left( x^{e - 1} + z^{e - 1} \right) + s^2 \left( x^{e - 2} + z^{e
  - 2} \right) + s^3 \left( x^{e - 3} + z^{e - 3} \right)   \pmod{s^4}  .
\end{eqnarray*}
Hence
\begin{multline} \label{eqlem1}
  s \left( x + z \right) \phi_e + \left( x + z \right)^2 \phi_e = \left( x^{e
  - 1} + z^{e - 1} \right) + s \left( x^{e - 2} + z^{e - 2} \right) + s^2
  \left( x^{e - 2} + z^{e - 2} \right) + \\ s^3 \left( x^{e - 3} + z^{e - 3}
  \right)  \pmod{s^4}  .
\end{multline}
As we have
\[ \left( x + z \right)^2 \phi_e = \left( x^{e - 1} + z^{e - 1} \right) 
  \pmod{s}, \]
and hence
\[ \left( x + z \right) \phi_e = \frac{x^{e - 1} + z^{e - 1}}{x + z}  \pmod{s}, \]
we deduce
\begin{eqnarray*}
  & \left( x + z \right)^2 \phi_e & = \left( x^{e - 1} + z^{e - 1} \right) +
  s \left( x^{e - 2} + z^{e - 2} \right) + s \left( x + z \right) \phi_e 
  \pmod{s^2}\\
  &  & = \left( x^{e - 1} + z^{e - 1} \right) + s \left( x^{e - 2} + z^{e -
  2} \right) + s \frac{x^{e - 1} + z^{e - 1}}{x + z}  \pmod{s^2}\\
  &  & = \left( x^{e - 1} + z^{e - 1} \right) + s \frac{x^{e - 2} z + z^{e -
  2} x}{x + z}  \pmod{s^2}.
\end{eqnarray*}
So we have
\begin{equation} \label{eqlem2}
 \left( x + z \right)^2 \phi_e = \left( x^{e - 1} + z^{e - 1} \right) + s
   \frac{x^{e - 2} z + z^{e - 2} x}{x + z}  \pmod{s^2} 
\end{equation}
and

\begin{equation} \label{eqlem3} \left( x + z \right) \phi_e = \frac{\left( x^{e - 1} + z^{e - 1}
   \right)}{x + z} + s \frac{x^{e - 2} z + z^{e - 2} x}{\left( x + z
   \right)^2}   \pmod{s^2} . \end{equation}
Using \ref{eqlem2} and \ref{eqlem3} in \ref{eqlem1} we get
\begin{eqnarray*}
  & \left( x + z \right)^2 \phi_e & = \left( x^{e - 1} + z^{e - 1} \right) +
  s \left( x + z \right) \phi_e + s \left( x^{e - 2} + z^{e - 2} \right) + s^2
  \left( x^{e - 3} + z^{e - 3} \right)  \pmod{s^3}\\
  &  & = \left( x^{e - 1} + z^{e - 1} \right) + s \frac{\left( x^{e - 1} +
  z^{e - 1} \right)}{x + z} + s^2 \frac{x^{e - 2} z + z^{e - 2} x}{\left( x +
  z \right)^2} + s \left( x^{e - 2} + z^{e - 2} \right) + \\
  &  & s^2 \left( x^{e - 3} + z^{e - 3} \right)  \pmod{s^3}\\
  &  & = \left( x^{e - 1} + z^{e - 1} \right) + s \frac{\left( x^{e - 2} z +
  z^{e - 2} x \right)}{x + z} + s^2 \frac{\left( x^{e - 3} + z^{e - 3} \right)
  \left( x^2 + z^2 + xz \right)}{\left( x + z \right)^2}  \pmod{s^3} .
\end{eqnarray*}

\end{proof}

\begin{lemma} \label{e1mod4}
Suppose $e \equiv 1 \pmod{4}$ and let $s=x+y$. We have :
\[
\left( x+z \right)^2 \phi_e = \left( x^{e-1}+z^{e-1} \right) + s\frac{ \left( x^{e-1}+z^{e-1} \right)}{x+z}+s^2\frac{\left( x^{e-1}+z^{e-1} \right) }{\left( x+z \right)^2} \pmod{s^3}
\]
\end{lemma}

\begin{proof}
The proof of lemma \ref{e1mod4} is similar to the proof of lemma \ref{e3mod4}.
\end{proof}

\begin{lemma} \label{eodd}
  For all odd $e \in \mathbb{N}$ we have
  \[ \phi_e (x, z, z) = \frac{x^{e - 1} + z^{e - 1}}{(x + z)^2} \]
\end{lemma}

The proof is straightforward from previous lemma. It can also be found in \cite{Delgado2012}

For all $k \in \{0, 1, \ldots, d\}$ we have
\[ a_k \phi_k = \sum_{i = 0}^9 P_i Q_{k - i - 3} . \]

\textbf{Degree $d-3$}\\

We have
\[ \phi_d = A^3 \phi_e^4 = P_9 Q_{d - 12} . \]
As $P_9 = A^3$, we get $Q_{d - 12} = \phi_e^4$.\medskip

\textbf{Degree $d - 4$}\\

We have
\[ a_{d - 1} \phi_{d - 1} = P_9 Q_{d - 13} + P_8 Q_{d - 12} . \]
As $P_8 = A^2  (s_1^2 \text{tr} (c_1) + s_2 \text{tr} (c_4))$, it gives us
\[ a_{d - 1} \phi_{d - 1} = A^3 Q_{d - 13} + A^2 \phi_e^4  (s_1^2 \text{tr}
   (c_1) + s_2 \text{tr} (c_4)) . \]
By lemma \ref{eodd} $\phi_{d - 1}$ is not divisible by $A$, so $a_{d - 1} = 0$ and
\[ AQ_{d - 13} = \phi_e^4  (s_1^2 \text{tr} (c_1) + s_2 \text{tr} (c_4)) . \]
We know that $A$ is prime with $s_1^2 \text{tr} (c_1) + s_2 \text{tr} (c_4)$
because $(x + y)$ does not divide this polynomial, and $A$ does not divide
either $\phi_5^4$, which implies $Q_{d - 13} = P_8 = 0$ and $\text{tr} (c_1) =
\text{tr} (c_4) = a_{d - 1} = 0$.\\

\textbf{Degree $d - 5$}\\

We have
\[ a_{d - 2} \phi_{d - 2} = a_{d - 2}  (A \phi_{2 e - 1}^2) = P_9 Q_{d - 14} +
   P_8 Q_{d - 13} + P_7 Q_{d - 12} . \]
Knowing that $P_8 = Q_7 = 0$ we obtain
\[ a_{d - 2}  (A \phi_{2 e - 1}^2) = P_9 Q_{d - 14} + P_7 Q_{d - 12} . \]
We also know that
\[ P_7 = A (s_1^4 q_1 (c_1) + s_2^2 q_1 (c_4) + s_1^2 s_2 q_5 (c_1, c_4)) +
   A^2 s_1 \text{tr} (b_1), \]
denoting

$q_1 (c_i) = c_i \rho (c_i) + c_i \rho^2 (c_i) + \rho (c_i) \rho^2 (c_i)$ and

$q_5 (c_1, c_4) = c_1  (\rho (c_4) + \rho^2 (c_4)) + c_4  (\rho (c_1) + \rho^2
(c_1)) + \rho (c_1) \rho^2 (c_4) + \rho (c_4) \rho^2 (c_1)$.

So
\begin{equation}
  a_{d - 2} \phi_{2 e - 1}^2 = A^2 Q_{d - 14} + \phi_e^4  (s_1^4 q_1 (c_1) +
  s_2^2 q_1 (c_4) + s_1^2 s_2 q_5 (c_1, c_4) + As_1 \text{tr}
  (b_1)),
\end{equation}
Putting $y = z$ we have
\[ a_{d - 2} \left( \frac{x^{4 e - 4} + z^{4 e - 4}}{(x + z)^4} \right) +
   \left( \frac{x^{4 e - 4} + z^{4 e - 4}}{(x + z)^8} \right)  (q_1 (c_1) x^4
   + q_1 (c_4) z^4 + x^2 z^2 q_5 (c_1, c_4)) = 0, \]
hence we obviously have $q_5 (c_1, c_4) = 0$ and $q_1 (c_1) = q_1 (c_4) = a_{d
- 2}$. We do not assume that $y = z$ anymore.

We know from (4) that $A$ divides $a_{d - 2}  (\phi_{2 e - 1}^2 + \phi_e^4
(s_1^4 + s_2^2) )$, as it is a square, $A^2$ divides it too. Replacing
in (4) we get
\[ a_{d - 2}  (\phi_{2 e - 1}^2 + \phi_e^4 (s_1^4 + s_2^2) )^2 + A^2
   Q_{d - 14} = A \phi_e^4 s_1 \text{tr} (b_1), \]
so $A$ divides $\text{tr} (b_1) s_1 \phi_e^4$. But $A$ divides neither $s_1$
nor $\phi_e^4$ so $\text{tr} (b_1) = 0$. In conclusion we have
\begin{eqnarray*}
  & P_7 = q_1 (c_1)  (s_1^2 + s_2)^2 A = q_1 (c_1) A \phi_5^2 . &
  \text{and}\\
  & Q_{d - 14} = q_1 (c_1) \frac{\phi_{2 e - 1}^2 + \phi_e^4 \phi_5^2}{A^2} .
  & 
\end{eqnarray*}
\begin{lemma}
  The polynomial $Q_{d - 14} (x, z, z)$ is equal to zero.
\end{lemma}

\begin{proof}
  from lemma \ref{e1mod4} and \ref{e3mod4} we get, if either $e \equiv 3 \pmod{4}$ or $e \equiv
  1 \pmod{4}$:
  \begin{eqnarray*}
   Q_{d - 14} = \left( \frac{\left( \frac{x^{2 e - 2} + z^{2 e - 2}}{(x +
     z)^2} + s \left( \frac{x^{2 e - 2} + z^{2 e - 2}}{(x + z)^3} \right) +
     s^2 R_1 \right) }{A} \right)^2 + \\ \left( \frac{ \left( \frac{x^{2 e - 2} + z^{2 e - 2}}{(x + z)^4} + s^2 R_2 \right) ((x + z)^2 + s (x + z) + s^2)}{A} \right)^2 
  \end{eqnarray*}
  \[ = \frac{s}{(x + y)  (x + z)} R_3, \]
  hence $Q_{d - 14} (x, z, z) = 0$.
\end{proof}

\textbf{Degree $d - 6$}\\

We have
\[ a_{d - 3} \phi_{d - 3} = P_9 Q_{d - 15} + P_8 Q_{d - 14} + P_7 Q_{d - 13} +
   P_6 Q_{d - 12} = P_9 Q_{d - 15} + P_6 Q_{d - 12} . \]
We know that
\begin{multline*} P_6 = A^2 \text{tr} (d_1) + A (s_1^3 q_5 (c_1, b_1) + s_1 s_2 q_5 (c_1, b_1))
   + s_1^6 N (c_1) + s_1^4 s_2 q_4 (c_1, c_4) + \\ s_1^2 s_2^2 q_4 (c_4, c_1) +
   s_2^3 N (c_4) \end{multline*}
where
\[ N (a) = a \rho (a) \rho^2 (a) \text{which is the norm of a in } \mathbb{F}_q, \]
\[ q_4 (a, b) = a \rho (a) \rho^2 (b) + a \rho (b) \rho^2 (a) + b \rho (a)
   \rho^2 (a) \]
and
\[ q_5 (a, b) = a (\rho (b) + \rho^2 (b)) + b (\rho (a) + \rho^2 (a)) + \rho
   (a) \rho^2 (b) + \rho (b) \rho^2 (a), \]
for all $a, b$ in $\mathbb{F}_{q^3}$.

Making $y = z$ we get:
\[ a_{d - 3} \phi_{d - 3} (x, z, z) = P_6 (x, z, z) \phi_e^4 (x, z, z), \]
with
\[ P_6 (x, z, z) = (c_1 x^2 + c_4 z^2)  (\rho (c_1) x^2 + \rho (c_4) z^2) 
   (\rho^2 (c_1) x^2 + \rho^2 (c_4) z^2) . \]
As
\[ \phi_{d - 3} (x, z, z) = \frac{x^{d - 4} + z^{d - 4}}{(x + z)^2} \]
and
\[ \phi_e^4 (x, z, z) = \frac{x^{d - 4} + z^{d - 4}}{(x + z)^8}, \]
we have
\[ (c_1 x^2 + c_4 z^2)  (\rho (c_1) x^2 + \rho (c_4) z^2)  (\rho^2 (c_1) x^2 +
   \rho^2 (c_4) z^2) = a_{d - 3}  (x + z)^6 . \]
Hence $c_1 = c_4$.\\
Now we have 

\begin{equation} \label{qd15}
N(c_1)\left( \phi_{d_3}+\phi_5^3 \phi_e^4 \right) = A^3 Q_{d-15} +  \tmop{tr}(d_1)A^2 \phi_e^4 + q_5 (c_1,b_1) A \phi_5 s_1 \phi_e^4.
\end{equation}

One can verify with lemma \ref{e3mod4} and \ref{e1mod4} that $A^2$ divides $\phi_{d_3}+\phi_5^3 \phi_e^4$ and we obtain $q_5 (c_1,b_1)=0$ since $\phi_5 s_1 \phi_e^4$ is prime with $A$. Plugging the last result into \ref{qd15} and dividing the whole expression by $A^2$ we get 

\[
A Q_{d-15} = N(c_1) \frac{\left( \phi_{d_3}+\phi_5^3 \phi_e^4 \right)}{A^2} + \tmop{tr}(d_1) \phi_e^4.
\]
Putting $y=z$, we obtain 
\[
N(c_1) \frac{\left( \phi_{d_3}+\phi_5^3 \phi_e^4 \right)}{A^2}(x,z,z) = \tmop{tr}(d_1) \phi_e^4(x,z,z).
\]

Now either $\frac{\left( \phi_{d_3}+\phi_5^3 \phi_e^4 \right)}{A^2}(x,z,z)$ is different from $\phi_e^4(x,z,z)$ and $\tmop{tr}(d_1) = N(c_1) = 0$, or $\frac{\left( \phi_{d-3}+\phi_5^3 \phi_e^4 \right)}{A^2}(x,z,z)= \phi_e(x,z,z)$ and $\tmop{tr}(d_1) = N(c_1)$ but in both case we have $\tmop{tr}(d_1) = N(c_1)$.\\

\textbf{Degree $d - 7$}

We have
\begin{equation} \label{qd16}
 a_{d - 4} \phi_{d - 4} = P_9 Q_{d - 16} + P_8 Q_{d - 15} + P_7 Q_{d - 14} +
   P_6 Q_{d - 13} + P_5 Q_{d - 12},
\end{equation}
where
\[ P_5 = q_4 (c_1, b_1) s_1 \phi_5^2 + A (q_1 (b_1) s_1^2 + q_5 (c_1, d_1)
   \phi_5), \]
We know that $\phi_{d - 4} = A^7 \phi_{\frac{e - 1}{2}}$ so making again $y =
z$ enables us to obtain:
\[ 0 = P_5 (x, z, z) = q_4 (c_1, b_1)  (x (x^2 + z^2)) \]
and finally $q_4 (c_1, b_1) = 0$. Now \ref{qd16} becomes

\[
a_{d-4}A^7 \phi_{\frac{e - 1}{2}} = A^3 Q_{d-16} + q_1(c_1)A \phi_5^2 Q_{d-14} +  \left( q_1(b_1)s_1^2 + q_5(c_1,d_1) \phi_5 \right) A \phi_e^4.
\]
We divide this expression by $A$ and we put $y=z$ and it gives

\[
q_1(b_1)x^2 = q_5(c_1,d_1)(x^2+y^2),
\]

so $q_1(b_1) = q_5(c_1,d_1)=0 $. \\

\textbf{degree $d - 8$}

For this step we have:
\[ a_{d - 5} \phi_{d - 5} = P_9 Q_{d - 17} + P_8 Q_{d - 16} + P_7 Q_{d - 15} +
   P_6 Q_{d - 14} + P_5 Q_{d - 13} + P_4 Q_{d - 12} . \]
with
\[ P_4 = q_4 (b_1, c_1) s_1^2 \phi_5 + q_4 (c_1, d_1) \phi_5^2 + q_5 (b_1, d_1)
   As_1, \]
Putting $y = z$ we get:
\begin{multline*} a_{d - 5}  \frac{x^{d - 6} + z^{d - 6}}{(x + z)^2} = \frac{1}{(x + z)^8} 
   ((q_4 (b_1, c_1) + q_4 (c_1, d_1)) (x^d + x^4 z^{d - 4}) + \\ q_4 (b_1, c_1)
   (x^{d - 2} z^2 + x^2 z^{d - 2}) + q_4 (c_1, d_1) (x^{d - 4} z^4 + z^d)) . 
\end{multline*}
Putting on the same denominator we have
\[ a_{d - 5}  (x^{d - 6} z^6 + x^6 z^{d - 6}) = 0 \]
and then $a_{d - 5} = 0$, therefore $q_4 (b_1, c_1) = q_4 (c_1, d_1) = 0$\\

\textbf{Summary}

At this point we get the following system
\begin{eqnarray*}
  \left\{ \begin{array}{l}
    q_1 (b_1) = 0\\
    \text{tr} (b_1) = 0\\
    q_5 (c_1, b_1) = 0\\
    \text{tr} (c_1) = 0\\
    q_4 (c_1, b_1) = 0\\
    q_4 (b_1, c_1) = 0\\
    q_4 (c_1, d_1) = 0 \\
    q_5(c_1,d_1) = 0 \\
    \tmop{tr}(d_1)= N(c_1)
  \end{array} \right. &  & 
\end{eqnarray*}
Let us suppose that $c_1 \neq 0$. The linear system in $b_1, \rho (b_1),
\rho^2 (b_1)$ formed by the three first equations gives $b_1 = 0$. Indeed, the
determinant of this system is $(c_1 + \rho (c_1))  (\rho (c_1) + \rho^2 (c_1))
(\rho^2 (c_1) + c_1)$ can vanish only if $c_1 = 0$ because $\text{tr} (c_1) =
0$.

If, moreover, $c_1 \neq \rho (c_1)$, the last 3 equations form a linear system
in $d_1, \rho (d_1), \rho^2 (d_1)$ which can gives
\[ d_1 = c_1^3 . \]
Therefore $R = c_1 \phi_5^2 + c_1^3$ which is the form given in the
proposition 4.

If $c_1 = \rho (c_1)$ then, as $\text{tr} (c_1) = 0$, $c_1 = 0$. Let us
suppose from now on that it is the case. We need to use
\[ a_{d - 6} \phi_{d - 6} = P_9 Q_{d - 18} + P_8 Q_{d - 17} + P_7 Q_{d - 16} +
   P_6 Q_{d - 15} + P_5 Q_{d - 14} + P_4 Q_{d - 13} + P_3 Q_{d - 12}, \]
when we replace $c_1$ by zero we get
\[ a_{d - 6} A \phi^2_{2 e - 1} = A^3 Q_{d - 18} + P_3 \phi_e^4, \]
where
\[ P_3 = N (b_1) s_1^3 + q_1 (d_1) A. \]
If moreover we make $y = z$ we obtain
\[ 0 = P_3 (x, z, z) = N (b_1) x^3 . \]
so $N (b_1) = 0$. Therefore $b_1 = 0$.

We now use
\begin{multline*} a_{d - 9} \phi_{d - 9} = P_9 Q_{d - 21} + P_8 Q_{d - 20} + P_7 Q_{d - 19} +
   P_6 Q_{d - 18} + P_5 Q_{d - 17} + P_4 Q_{d - 16} + P_3 Q_{d - 15} + \\ P_2
   Q_{d - 14} + P_1 Q_{d - 13} + P_0 Q_{d - 12},
\end{multline*}
which gives:
\[ a_{d - 9} \phi_{d - 9} = A^3 Q_{d - 21} + N (d_1) \phi_e^4. \]
If we put $y = z$ we obtain
\[ a_{d - 9}  \frac{x^{d - 10} + z^{d - 10}}{(x + z)^2} = N (d_1) \frac{x^{d -
   4} + z^{d - 4}}{(x + z)^8}. \]
Putting on the same denominator we get $a_{d - 9} = 0$
and therefore $N (d_1) = 0$, hence $d_1 = 0$. It means that $R = 0$, finally
proving the first part of proposition 1.\\

Now let us consider $L \left( x \right) = x \left( x + c_1 \right) \left( x + \rho
\left( c_1 \right) \right) \left( x + \rho^2 \left( c_1 \right) \right)$,
since $\tmop{tr} \left( c_1 \right) = 0$, $L$ is a $q$-affine polynomial and as $L(x)$ has only one root of $0$ in $\mathbb{F}_q$ (that is $x=0$), $L(x)$ is a $q$-affine permutation. One can verify that
\[ \frac{L \left( x \right)^3 + L \left( y \right)^3 + L \left( z \right)^3 +
   L \left( x + y + z \right)^3}{\left( x + y \right) \left( y + z \right)
   \left( z + x \right)} = \left( A + R \right) \left( A + \rho \left( R
 \right) \left( A + \rho^2 ( R \right) \right) . \]
So it means that the polynomial $\phi$ associated to $L \left( x \right)^3$ divides $\phi_f$, which proves the second part of proposition 1.

\end{proof}

We can now complete the proof of theorem 9 by showing that $f$ is CCZ-equivalent to a polynomial of degree $e$.

\section{CCZ-equivalence}
\label{sec:4}

Let us consider $c_1 \in \mathbb{F}_{q^3}$ such that $\text{tr}(c_1)=0$ and $R \left( x, y, z \right) = c_1 \phi_5 +
c_{1}^3, \in \mathbb{F}_{q^3}\left[ x, y, z \right]$. We recall that $L \left( x \right) = x \left( x + c_1 \right) \left( x +
\rho \left( c_1 \right) \right) \left( x + \rho^2 \left( c_1 \right) \right)$.

\begin{theorem}
  Let $f$ be a function such that $\deg \left( f \right) = 4 e$, with $e > 3$ such that $\phi_e$ is absolutely irreducible, and such that the polynomials of the form
  \[ \left( x + y \right) \left( x + z \right) \left( y + z \right) + R, \]
  divides $\phi$, therefore $f$ is CCZ-equivalent to $x^e + S \left( x
  \right)$, where $S \in \mathbb{F}_q \left[ x \right]$ is of degree at most $e - 1$.
\end{theorem}

\begin{proof}

  Let us consider the set $G$ of the polynomials of the form $g \left( x
  \right) = L \left( x \right)^e + S \left( L \left( x \right) \right)$, where
  $S$ is a polynomial of $\mathbb{F}_q \left[ x \right] $ of degree at most $e - 1$ with no
  monomials of exponent a power of 2. Let $\delta$ be the number of power of 2
  less or equal than $e - 1$. It is easy to remark that $G$ defines an affine
  subspace of the vector space $\mathbb{F}_q [x]$ of dimension $e - \delta$.
  We denote by $\phi_g$ the polynomial $\phi$ associated to $g$ and
  $\phi_{L^n}$ the polynomial $\phi$ associated to $L^n$. So we have
  \[ \phi_g = \phi_{L^{^e}} + S \left( \phi_{L^i} \right) . \]

  Now let us consider the set $F$ of all the polynomials $f$ of degree $4 e$
  with leading coefficient $1$ such that $\phi_{L^3}$ divides their associated
  polynomials $\phi$ and such that $f$ does not have any monomial of exponent
  a power of 2. The goal of this proof is to show that $F = G$. We begin by
  proving that $G \subset F$, then we show that they have the same dimension.
  
  \begin{lemma}
    The set $G$ is a subset of $F$.
  \end{lemma}
  
  \begin{proof}
    It is sufficient to prove that $\phi_{L^3}$ divides $\phi_{L^n}$ for all
    $n \geqslant 3$.

    We know that $x^3 + y^3 + z^3 + \left( x + y + z \right)^3 = A$ divides
    $x^n + y^n + z^n + \left( x + y + z \right)^n$. Putting
    \begin{eqnarray*}
      & X = L \left( x \right) & \\
      & Y = L \left( y \right) & \\
      & Z = L \left( z \right) & 
    \end{eqnarray*}
    we have $X^3 + Y^3 + Z^3 + \left( X + Y + Z \right)^3$ divides $X^n + Y^n
    + Z^n + \left( X + Y + Z \right)^n$. As $\text{tr}(c_1)=0$, $L \left( x \right)$ is a linearized polynomial 
    so $X + Y + Z = L \left( x \right) + L \left( y \right) + L \left( z \right) = L \left( x + y + z
    \right)$ therefore $L \left( x \right)^3 + L \left( y \right)^3 + L \left(
    z \right)^3 + L \left( x + y + z \right)^3$ divides $L \left( x \right)^n
    + L \left( y \right)^n + L \left( z \right)^n + L \left( x + y + z
    \right)^n$ then $\phi_{L^3}$ divides $\phi_{L^n}$.
  \end{proof}
  
  \begin{lemma}
    $F$ defines an affine subspace of the vector space $\mathbb{F}_q \left[ x
    \right]$ of dimension less or equal than $e - \delta$.
  \end{lemma}
  
  \begin{proof}
  We consider the mapping:
    \begin{eqnarray*}
      & \varphi : & F \rightarrow \mathbb{F}_q^{e - \delta}\\
      &  & f \rightarrow \left( a_{d - 4}, \ldots, a_{12} \right)
    \end{eqnarray*}
    It is sufficient to prove that this mapping is one-to-one.
    
    Let $f$ and $f'$ in $F$ be two elements such that $\varphi \left( f
    \right) = \varphi \left( f' \right)$. We write $f = \sum_{i = 0}^d a_i
    x^i$ and $f' = \sum_{i = 0}^d a'_i x^i$. We note $a_k \phi_k = \sum_{i =
    0}^9 P_i Q_{k - i - 3}$ and $a_k' \phi_k = \sum_{i = 0}^9 P_i Q'_{k - i -
    3} .$

    We will show by induction that $a_i = a'_i$ for all $0 \leqslant i
    \leqslant d$ and that $Q_i = Q_i'$ for all $0 \leqslant i \leqslant d -
    12$.
    
    We have $a_d = a_d' = 1$ and $Q_{d - 12} = Q_{d - 12}' = \phi_e^4$.
    
    Suppose that $a_j = a'_j$ and that $Q_{j - 12} = Q'_{j - 12}$ for $j > i$.
    Let us show that $a_i = a'_i$ and $Q_{i - 12} = Q'_{i - 12}$ if $4$ does
    not divide $i$.
    
    If $i \geqslant 12$, we have
    \[ a_i \phi_i = \sum^9_{\sup \left( 0, i-d + 9 \right)} P_k Q_{i-k-3} =
       A^3 Q_{i-12} + \sum^8_{\sup \left( 0, i-d + 9 \right)} P_k Q_{i-k-3},
    \]
    so $A^3$ divides
    \[ a_i \phi_i + \sum^8_{\sup \left( 0, i-d + 9 \right)} P_k Q_{i-k-3} .
    \]
    It divides
    \[ a_i' \phi_i + \sum^8_{\sup \left( 0, i-d + 9 \right)} P_k Q'_{i-k-3} =
       a_i' \phi_i + \sum^8_{\sup \left( 0, i-d + 9 \right)} P_k Q_{i-k-3}, \]
    because $i-k-3 \geqslant i-11$. So it divides $\left( a_i + a'_i
    \right) \phi_i$. If $4$ does not divide $i$ then $A^3$ does not divide
    $\phi_i$ so $a_i = a'_i$ and
    \[ Q_{i-12} = \frac{a_i \phi_i + \sum^8_{\sup \left( 0, i-d + 9 \right)}
       P_k Q_{i-k-3}}{A^3} = \frac{a_i' \phi_i + \sum^8_{\sup \left( 0, i-d +
       9 \right)} P_k Q'_{i-k-3}}{A^3} = Q_{d-12}' . \]

  \end{proof}
  
  From lemma 5 and 6 we obtain $F = G$. So every $f \in F$ is of the form $L \left( x \right)^e + S \left( L \left(x \right) \right)$ and hence they are CCZ-equivalent to $x^{^e} + S \left( x\right)$. If $f$ is of degree $4 e$ with leading coefficient $1$ such that $\phi_{L^3}$ divides their associated polynomials $\phi$ and has monomials of exponent a power of $2$, then $f$ is CCZ-equivalent to a polynomial in $F$ therefore it is also CCZ-equivalent to \ $x^{^e} + S \left( x \right)$.
\end{proof}

We now have that $f$ is CCZ-equivalent to a polynomial of degree $e$ which is
odd. As $e$ is odd and not a Gold or Kasami number (see remark 2), we can deduce from theorem 1 that $f$ cannot be an Exceptional APN function. Contradiction.

\bibliographystyle{spmpsci}      % mathematics and physical sciences

\end{document}